\documentclass[preprint]{revtex4-2}
\usepackage{amsmath,amssymb,amsthm}
\usepackage{graphicx}
\usepackage{placeins}
\usepackage{microtype}
\usepackage{booktabs}
\usepackage[ruled,linesnumbered]{algorithm2e}
\usepackage{hyperref}

\newtheorem{theorem}{Theorem}

\newtheorem{proposition}{Proposition}

\newcommand{\E}{\mathbb{E}}
\newcommand{\cO}{\mathcal{O}}
\newcommand{\cF}{\mathcal{F}}
\newcommand{\PP}{\mathbb{P}}
\newcommand{\R}{\mathbb{R}}
\newcommand{\one}{\mathbf{1}}
\newcommand{\Tr}{\operatorname{Tr}}
\newcommand{\Diss}[1]{\mathcal D[#1]}
\newcommand{\Gain}[1]{\mathcal G[#1]}

\begin{document}

\title{Collapse time under weak measurement on $N$ entangled sites}

\newcommand{\cofirstnote}{These authors contributed equally to this
  work. Email addresses: tsvan@hcmut.edu.vn (T.-S.~P.~Van),
weichao.liang@xjtu.edu.cn (W.~Liang).}

\author{Truong-Son P. Van}
\thanks{\cofirstnote}
\affiliation{Ho Chi Minh City University of Technology, Vietnam
National University Ho Chi Minh City, Vietnam}

\author{Weichao Liang}
\thanks{\cofirstnote}
\affiliation{School of Automation Science and Engineering, Xi'an
Jiaotong University, Xi'an, China}

\author{David W. Snoke}
\email{snoke@pitt.edu}
\affiliation{Department of Physics and Astronomy, University of
Pittsburgh, and Pittsburgh Quantum Institute, Pittsburgh, PA, 15260, USA}

\date{\today}

\begin{abstract}
  We determine how long local weak measurements take to
  concentrate a single excitation on one of $N$ sites in the model of
  Phys. Rev. A \textbf{111}, 032217 (2025). For equal initial occupation
  probabilities and a fixed $0<\delta<1/2$, let $T_N$ be the first time
  any occupation probability reaches $1-\delta$. Measuring time in
  units of the characteristic measurement time $\tau_m$ of each
  detector, we prove
  $\E|T_N-\tfrac12\ln N|^r\le C_{r,\delta}(\ln N)^{r/2}$
  for every $r>0$, with $C_{r,\delta}$ independent of $N$.
  Consequently, $T_N/\ln N\to1/2$  in every $L^r$,
  $r>0$. This result corrects the previously proposed double-logarithmic
  asymptotic scaling.
  We also propose a new algorithm that simulates the system more efficiently.
\end{abstract}

\maketitle

\section{Introduction and main results}
\label{sec:intro}

Nonlocal collapse of an extended wave function impinging on multiple
detectors is one of the most striking effects of quantum mechanics.
For example, a plane wave hitting a cloud chamber with billions of
atoms will create distinct particle tracks, even if the initial state
is a superposition of all the atoms being excited; similarly, in the
canonical single-photon two-slit interference experiment, one
photographic film grain changes color at a time, even though an
extended wave state has initially put millions of film grains into a
superposition of each detecting a photon.

Prior work \cite{VJS} examined the time scale for this collapse,
namely the time for a macroscopic superposition of millions to
billions of atoms in excited states to have just one of them in a
definite excited state, from the perspective of weak measurement
theory and/or spontaneous nonlocal collapse theory. In practice, if
the time scale for this macroscopic collapse is very long, for
example, scaling linearly with the number $N$ of atoms in the
superposition, the scenario would be indistinguishable from an
Everett-Wheeler many-worlds scenario, in which many macrosopic
``cats'' coexist for indefinite periods of time. (For a review of
quantum interpretations and the many-worlds approach, see Ref.~\cite{snokeQM}.)

The model of Ref.~\cite{VJS} uses non-unitary noise to drive collapse.
In the Copenhagen interpretation, this noise can be viewed as the
residual effect of weak measurements by ``observers'' detecting the
environment \cite{jordanbook}. In the spontaneous collapse
interpretation, it is a physical, nonlocal noise term. (For an
  overview of the latter approach, see Ref.~\cite{snokecollapse3} and
references therein.) Both interpretations give the same noise term.
Over the simulated range in Ref.~\cite{VJS}, the collapse time grew
slowly with system
size and remained a few times the single-detector collapse time.

The asymptotic scaling law, however, remained unclear. Over a wide
range of $N$, the simulations were well approximated by $C \ln \ln N$
for some constant $C$. We prove here that the leading asymptotic law is
$\frac{1}{2}\ln N$ in units of the characteristic measurement time.
For example, if all the atoms in the visible universe, of order
$10^{80}$, were entangled, this leading estimate would be about 128
times the two-site mean collapse time $\langle T_2\rangle$.
If this reference time corresponds to a typical coherence time,
ranging from 1 ps in solid-state detectors to 1 ns in gases
\cite{gas}, the estimate remains below a millisecond for the entire
visible universe.

We consider a single excitation shared among $N\ge2$ sites as in
Ref.~\cite{VJS}. For the collapse-time theorem, we take equal initial
weights $|\alpha_n(0)|^2=1/N$. Each site is continuously monitored by its own
detector, and there is no Hamiltonian dynamics. We ask how long these
local measurements take to concentrate the occupation probability on
one site, and how this time depends on $N$.  Because the detectors
are nonlocally entangled, the collapse time does not just depend on
local properties. A weak measurement that increases the
occupation probability at one site makes it {\em less} likely that
the excitation is on
one of the other atoms, no matter how far away.

\FloatBarrier
\section{Derivation of the collapse time formula}

\subsection{The stochastic model and the collapse time}
\label{sec:model}

We first recall the stochastic model in Ref.~\cite{VJS}. Let $|n\rangle$
denote the state with the excitation at site $n$ and all other sites
unoccupied. The conditional quantum state has the form
\begin{equation*}
  |\psi(t)\rangle=\sum_{n=1}^N\alpha_n(t)|n\rangle \,,
\end{equation*}
where
\begin{equation*}
  \sum_{n=1}^N|\alpha_n(t)|^2=1 \,.
\end{equation*}

We express each occupation probability through the local Bloch $z$-coordinate
\begin{equation*}
  U_n(t):=2|\alpha_n(t)|^2-1.
\end{equation*}
Here, $U_n=1$ means that the excitation is at site $n$, while $U_n=-1$
means that site $n$ is unoccupied.

As in Ref.~\cite{VJS}, $\tau_m$ denotes the characteristic measurement
time of each detector. We assume identical detectors and measure time
in units of $\tau_m$ throughout. Let $p_n=|\alpha_n(0)|^2$ be the
initial weight at site $n$. The model derived in Ref.~\cite{VJS} is
\begin{equation}
  dU_n=(1-U_n^2)\,dB_n
  -\sum_{k\ne n}(1+U_n)(1+U_k)\,dB_k,
  \qquad U_n(0)=2p_n-1,
  \label{eq:original-model}
\end{equation}
where $B_1,\dots,B_N$ are independent standard
Brownian motions.
The first term describes the effect of the detector
at site $n$. The sum describes the effect of the other detectors.

For the collapse-time problem, it is convenient to work directly with
\begin{equation*}
  w_n(t):=|\alpha_n(t)|^2=\frac{1+U_n(t)}2,
  \qquad w_n(0)=p_n.
\end{equation*}
Then $\sum_nw_n=1$, and the weights satisfy
\begin{equation}
  \label{eq:weight-sde}
  dw_n
  = 2w_n(1-w_n)\,dB_n-2w_n\sum_{k\ne n}w_k\,dB_k \,.
\end{equation}
Write $w=(w_1,\ldots,w_N)$ and $B=(B_1,\ldots,B_N)$, and let $e_n$
be the $n$th standard basis vector in $\R^N$. The vector form of the
weight equation is
\begin{equation}
  dw_n=2w_n(e_n-w)\cdot dB \,.
  \label{eq:weight-sde-vector}
\end{equation}
Appendix~\ref{app:sme} gives another derivation of
\eqref{eq:weight-sde} from the stochastic master equation
of Ref.~\cite{VJS}.

We fix a threshold $\delta\in(0,\tfrac12)$, and write
\begin{equation*}
  M(t):=\max_nw_n(t) \,.
\end{equation*}
We define a collapse as the first time $M(t)$ reaches $1-\delta$:
\begin{equation}
  T_N:=\inf\bigl\{t\ge0:M(t)\ge1-\delta\bigr\}.
  \label{eq:collapse-time}
\end{equation}

Our main goal is to show that $\E T_N \sim \frac12 \ln N$ as $N\to \infty$.
\begin{theorem}[Collapse time]
  \label{thm:main}
  Suppose $w_n(0)=1/N$ for every $n$. For every $0<\delta<1/2$
  and every $r>0$ there is a constant
  $C_{r,\delta}$, independent of $N$, such that
  \begin{equation}
    \E\Bigl|T_N-\frac12\ln N\Bigr|^r\le C_{r,\delta}\,(\ln N)^{r/2},
    \qquad N\ge2.
    \label{eq:moments}
  \end{equation}
  Consequently, $T_N$ is almost surely finite, and $T_N/\ln N\to1/2$ in
  probability and in $L^r$ for every $r>0$.
\end{theorem}

\subsection{Proof of the collapse time theorem}
\label{sec:proof}

We first derive the Brownian record representation for any fixed initial weights
$p_n=w_n(0)\ge 0$ with $\sum_n p_n=1$. For random initial weights,
we condition on their realized values.
The weights sum to one, and the vector form \eqref{eq:weight-sde-vector} gives
\begin{equation*}
  d \langle w_n \rangle_t = 4 w_n^2 \| e_n - w \|^2 \, dt \,.
\end{equation*}
For $p_n>0$, It\^o's formula gives, up to the first time $w_n$ reaches zero,
\begin{align*}
  d \log w_n &= \frac{ d w_n}{w_n} - \frac{ d \langle w_n \rangle_t}{
  2 w_n^2 } \\
  &= 2 (e_n - w) \cdot d B - 2 \| e_n - w \|^2 \, dt \,.
\end{align*}
Integrating gives
\begin{equation}
  w_n(t)=p_n\exp\left[
    2\int_0^t(e_n-w(s))\cdot dB(s)
  -2\int_0^t\lVert e_n-w(s)\rVert^2\,ds\right] \,.
  \label{eq:positive-weights}
\end{equation}
The exponent is finite on every finite time interval, so a weight with
$p_n>0$ stays positive. If $p_n=0$, the scalar linear equation
\eqref{eq:weight-sde-vector} gives $w_n(t)=0$ for all $t$, as in
\eqref{eq:positive-weights}.
Each $w_n$ is a bounded martingale.

For equal initial weights and $N>2$, Theorem~1 in Appendix~B of
Ref.~\cite{VJS} gives
$\E[w_i(t)w_j(t)]\to0$ for every $i\ne j$.
For any initial weights, the pairwise calculation in
Appendix~\ref{app:exponential-reduction} gives
$\E\sqrt{w_i(t)w_j(t)}=e^{-t}\sqrt{p_i p_j}$ for $i\ne j$.
Since $w_iw_j\le\sqrt{w_iw_j}$, the convergence holds for every $N\ge2$.
Each bounded martingale $w_n$ has an almost-sure and $L^1$ limit.
By bounded convergence, the products of distinct limiting weights
have expectation zero and hence vanish almost surely. Since the
limiting weights sum to one, exactly one of them equals one.
Thus, there is a random site $I\in\{1,\ldots,N\}$ such that
\begin{equation}
  w_n(t)\longrightarrow\one_{\{I=n\}}
  \quad\hbox{almost surely and in }L^1 \,.
  \label{eq:eventual-site}
\end{equation}
Physically, $I$ is the site that remains excited at $t=\infty$.
Equation~\eqref{eq:eventual-site} implies $T_N<\infty$ almost surely.
The martingale property gives
\begin{equation}
  \label{eq:eventual-prob}
  \PP(I=i) = \E w_i(t)=p_i
\end{equation}
and
\begin{equation}
  w_n(t)= \E \left[ \one_{\{I=n\}} \mid \cF_t  \right] =
  \PP(I=n\mid\cF_t) \,,
  \label{eq:posterior}
\end{equation}
where $\cF_t$ is the augmented filtration of the Brownian
motions.

We define the Brownian records by
\begin{equation}
  Y_n (t):=B_n(t)+\int_0^t(2w_n(s)-1)\,ds \,.
  \label{eq:Y-definition}
\end{equation}
For $p_i,p_j>0$, take the ratio in \eqref{eq:positive-weights}. Since
$\|e_i-w\|^2-\|e_j-w\|^2=-2(w_i-w_j)$, we obtain
\begin{align*}
  \ln\frac{w_i(t)}{w_j(t)}-\ln\frac{p_i}{p_j}
  &=2\bigl(B_i(t)-B_j(t)\bigr)
  +4\int_0^t\bigl(w_i(s)-w_j(s)\bigr)\,ds\\
  &=2\bigl(Y_i(t)-Y_j(t)\bigr).
\end{align*}
Thus $w_n(t)$ is proportional to $p_n e^{2Y_n(t)}$ on the sites
with $p_n>0$. The other weights remain zero. Since $\sum_j w_j(t)=1$,
\begin{equation}
  w_n(t)=\frac{p_n e^{2Y_n(t)}}{\sum_j p_j e^{2Y_j(t)}}.
  \label{eq:softmax}
\end{equation}
For equal initial weights, the prefactor in \eqref{eq:positive-weights}
is $1/N$. This common factor cancels in \eqref{eq:softmax}, giving
$w_n(t)=e^{2Y_n(t)}/\sum_j e^{2Y_j(t)}$.

Note that $Y_n$'s are adapted to $\cF_t$.
The identity
$B_n(t)=Y_n(t)-\int_0^t(2w_n(s)-1)\,ds$
then recovers the Brownian motions.
Hence $\cF_t^Y = \cF_t$.

For each $\ell$ with $p_\ell>0$, write
$\PP^\ell=\PP(\,\cdot\mid I=\ell)$ and let $\E^\ell$ denote
expectation under $\PP^\ell$.
For every event $A\in\cF_t$, the definition of conditional
probability and~\eqref{eq:eventual-prob} give
\begin{equation*}
  \PP^\ell(A)
  =\frac{\PP(A\cap\{I=\ell\})}{\PP(I=\ell)}
  =\frac{1}{p_\ell}\E\!\left[\one_A\one_{\{I=\ell\}}\right] \,.
\end{equation*}
Since $\one_A$ is $\cF_t$-measurable, the tower property gives
\begin{equation*}
  \begin{aligned}
    \E[\one_A\one_{\{I=\ell\}}]
    &=\E\!\left[\E[\one_A\one_{\{I=\ell\}}\mid\cF_t]\right]\\
    &=\E\!\left[\one_A\E[\one_{\{I=\ell\}}\mid\cF_t]\right]\\
    &=\E[\one_Aw_\ell(t)] \,,
  \end{aligned}
\end{equation*}
where the last equality is the posterior identity \eqref{eq:posterior}.
Thus $\PP^\ell(A)=\E[\one_A w_\ell(t)/p_\ell]$ for every
$A\in\cF_t$, which identifies the conditional density:
\begin{equation}
  \left.\frac{d\PP^\ell}{d\PP}\right|_{\cF_t}
  =Z_t^\ell:=\frac{w_\ell(t)}{p_\ell}.
  \label{eq:conditional-density}
\end{equation}

Now, we can write the equation for $Z^\ell$.
\begin{equation*}
  \begin{aligned}
    dZ_t^\ell
    &=\frac{1}{p_\ell}\,dw_\ell(t)\\
    &=\frac{2w_\ell(t)}{p_\ell}\left(dB_\ell(t)
    -\sum_{k=1}^Nw_k(t)\,dB_k(t)\right)\\
    &=2Z_t^\ell\left(dB_\ell(t)
    -\sum_{k=1}^Nw_k(t)\,dB_k(t)\right).
  \end{aligned}
\end{equation*}

Note that
\begin{equation*}
  dB_\ell(t)=\sum_{k=1}^N\one_{\{k=\ell\}}\,dB_k(t).
\end{equation*}
We therefore obtain
\begin{equation*}
  dZ_t^\ell
  =2Z_t^\ell\sum_{k=1}^N
  \bigl(\one_{\{k=\ell\}}-w_k(t)\bigr)\,dB_k(t).
\end{equation*}
The coefficients $2(\one_{\{k=\ell\}}-w_k(t))$ are bounded, so
Novikov's condition holds on every finite interval. Girsanov's
theorem shows that
\begin{equation}
  B_k^\ell(t):=B_k(t)
  -2\int_0^t\bigl(\one_{\{k=\ell\}}-w_k(s)\bigr)\,ds
  \label{eq:conditional-brownian}
\end{equation}
are independent standard Brownian motions under $\PP^\ell$.
Substitution into \eqref{eq:Y-definition} gives
\begin{equation}
  Y_k(t)=(2\one_{\{k=\ell\}}-1)t+B_k^\ell(t).
  \label{eq:conditional-records}
\end{equation}
For the rest of the proof of Theorem~\ref{thm:main}, take $p_n=1/N$
for all $n$. Then $\PP(I=\ell)=1/N$ and $Z_t^\ell=Nw_\ell(t)$.
Substituting \eqref{eq:conditional-records} into \eqref{eq:softmax} gives
\begin{equation}
  w_i(t)=
  \frac{\exp(4t\,\one_{\{i=\ell\}}+2B_i^\ell(t))}
  {\sum_j\exp(4t\,\one_{\{j=\ell\}}+2B_j^\ell(t))}
  \quad\PP^\ell\hbox{-almost surely}.
  \label{eq:conditional-softmax}
\end{equation}

\bigskip

Let $\phi(x)=(2\pi)^{-1/2}e^{-x^2/2}$,
$\Phi(x)=\int_{-\infty}^x\phi(v)\,dv$, and
$\overline\Phi=1-\Phi$.
\begin{proposition}[Early crossing]
  \label{prop:early-records}
  For every $t>0$ and $x\in\R$,
  \begin{equation}
    \PP(T_N\le t)
    \le\frac{\overline\Phi(x-2\sqrt t)+\Phi(x)^{N-1}}
    {1-\delta}.
    \label{eq:early-records}
  \end{equation}
\end{proposition}
\begin{proof}
  Since every $w_i$ is a martingale, their maximum $M$ is a
  nonnegative submartingale. Indeed, for $0\le s\le t$,
  \begin{equation*}
    \E[M(t)\mid\cF_s]
    \ge\max_i\E[w_i(t)\mid\cF_s]
    =M(s).
  \end{equation*}
  Doob's maximal inequality gives
  \begin{equation}
    \PP(T_N\le t)
    =\PP\!\left(\sup_{0\le s\le t}M(s)\ge1-\delta\right)
    \le\frac{\E M(t)}{1-\delta}.
    \label{eq:doob-maximum}
  \end{equation}

  Now, let
  \begin{equation*}
    \widehat I_t:=\min\{i:w_i(t)=M(t)\}.
  \end{equation*}
  This index is $\cF_t$-measurable. By \eqref{eq:posterior},
  \begin{equation*}
    \PP(I=\widehat I_t\mid\cF_t)
    = \sum_i \E \left[ \one_{\left\{ \widehat I_t = i  \right\}}
      \one_{\left\{ I = i
    \right\}} \mid \cF_t \right]
    =\sum_i\one_{\{\widehat I_t=i\}}w_i(t)=M(t).
  \end{equation*}
  Taking expectations and using the tower property gives
  \begin{equation*}
    \E M(t)
    =\E\!\left[\PP(I=\widehat I_t\mid\cF_t)\right]
    =\PP(\widehat I_t=I).
  \end{equation*}
  Thus $\E M(t)$ is the probability that the site with the largest
  current weight is the eventual occupied site. We split this event
  according to the possible values of $I$:
  \begin{equation}
    \begin{aligned}
      \E M(t)
      &=\PP(\widehat I_t=I)\\
      &=\sum_{\ell=1}^N\PP(\widehat I_t=\ell,I=\ell)\\
      &=\sum_{\ell=1}^N\PP(I=\ell)\,
      \PP(\widehat I_t=\ell\mid I=\ell)\\
      &=\frac1N\sum_{\ell=1}^N\PP^\ell(\widehat I_t=\ell).
    \end{aligned}
    \label{eq:MAP-identity}
  \end{equation}
  The last equality follows from~\eqref{eq:eventual-prob} and the definition
  $\PP^\ell=\PP(\,\cdot\mid I=\ell)$.

  Fix $t>0$ and $i\ne j$. By \eqref{eq:conditional-records},
  \begin{equation*}
    Y_i(t)-Y_j(t)
    =2t\bigl(\one_{\{i=\ell\}}-\one_{\{j=\ell\}}\bigr)
    +B_i^\ell(t)-B_j^\ell(t).
  \end{equation*}
  Under $\PP^\ell$, the two Brownian values are independent
  Gaussian variables, each with mean zero and variance $t$.
  Their difference therefore has mean zero and variance $t+t=2t$.
  It follows that
  \begin{equation*}
    Y_i(t)-Y_j(t)\sim
    \mathcal N\!\left(
      2t\bigl(\one_{\{i=\ell\}}-\one_{\{j=\ell\}}\bigr),2t
    \right).
  \end{equation*}
  Since $t>0$, this distribution has a density, and hence
  $\PP^\ell(Y_i(t)=Y_j(t))=0$.
  Furthermore,
  from \eqref{eq:conditional-softmax},
  for $k\ne\ell$,
  \begin{equation*}
    w_\ell(t)>w_k(t)
    \quad\Longleftrightarrow\quad
    4t+2B_\ell^\ell(t)>2B_k^\ell(t)
    \quad\PP^\ell\hbox{-almost surely}\,.
  \end{equation*}
  It then follows that
  \begin{equation*}
    \{\widehat I_t=\ell\}
    =\left\{4t+2B_\ell^\ell(t)>
    \max_{k\ne\ell}2B_k^\ell(t)\right\}
    \quad\PP^\ell\hbox{-almost surely}.
  \end{equation*}
  Dividing by $2\sqrt t>0$ and setting
  $\xi_k=B_k^\ell(t)/\sqrt t$, we obtain
  \begin{equation}
    \label{eq:Ihat}
    \{\widehat I_t=\ell\}
    =\left\{2\sqrt t+\xi_\ell>
    \max_{k\ne\ell}\xi_k\right\}
    \quad\PP^\ell\hbox{-almost surely},
  \end{equation}
  where the $\xi_k$ are independent standard normal variables
  under $\PP^\ell$.
  By symmetry, $\PP^\ell(\widehat I_t=\ell) = \PP^k(\widehat I_t= k)$
  for all $\ell, k$. Therefore, for any $x\in\R$,
  \eqref{eq:MAP-identity}  gives
  \begin{align*}
    \E M(t)
    &=\frac1N\sum_{\ell=1}^N\PP^\ell(\widehat I_t=\ell)\\
    &=\PP^1\!\left(2\sqrt t+\xi_1>
    \max_{2\le j\le N}\xi_j\right)\\
    &\le\PP^1(2\sqrt t+\xi_1\ge x)
    +\PP^1\!\left(\max_{2\le j\le N}\xi_j\le x\right)\\
    &=\overline\Phi \left( x-2\sqrt t \right)+\Phi(x)^{N-1}.
  \end{align*}
  The last equality uses the independence of $\xi_2,\dots,\xi_N$.
  Combining this estimate with \eqref{eq:doob-maximum} proves
  \eqref{eq:early-records}.
\end{proof}

\begin{proposition}[Late crossing]
  \label{prop:late-gauss}
  For every $N\ge2$, every $0<\delta<1/2$ and every $u>0$,
  \begin{equation}
    \PP\Bigl(T_N>\frac12\ln N+u\Bigr)
    \le\exp\Bigl(-\frac{u^2}{4\ln N+8u}\Bigr)
    +\frac{1-\delta}{\delta}\,e^{-u}.
    \label{eq:late-gauss}
  \end{equation}
\end{proposition}
\begin{proof}
  To bound the late tail, we use the weight of the eventual site.
  We first show that, for every $t>0$ and $v\ge0$,
  \begin{equation}
    \PP(T_N>t)\le
    \overline\Phi\!\left(\frac{v}{2\sqrt t}\right)
    +\frac{1-\delta}{\delta}(N-1)e^{-2t+v}.
    \label{eq:late-records}
  \end{equation}
  Fix $\ell$. Since the weights sum to one,
  $1-w_\ell(t)=\sum_{k\ne\ell}w_k(t)$.
  Under $\PP^\ell$, from~\eqref{eq:conditional-softmax}, it is true that
  \begin{align*}
    \frac{1-w_\ell(t)}{w_\ell(t)}
    &=\sum_{k\ne\ell}\frac{w_k(t)}{w_\ell(t)}\\
    &=\sum_{k\ne\ell}
    \frac{e^{2B_k^\ell(t)}}{e^{4t+2B_\ell^\ell(t)}}\\
    &=e^{-4t-2B_\ell^\ell(t)}
    \sum_{k\ne\ell}e^{2B_k^\ell(t)}.
  \end{align*}
  If $T_N>t$, then $w_\ell(t)<1-\delta$, so
  $(1-w_\ell(t))/w_\ell(t)>\delta/(1-\delta)$.
  If also $2B_\ell^\ell(t)\ge-v$, the preceding identity gives
  \begin{equation*}
    \sum_{k\ne\ell}e^{2B_k^\ell(t)}
    >\frac{\delta}{1-\delta}e^{4t+2B_\ell^\ell(t)}
    \ge\frac{\delta}{1-\delta}e^{4t-v}.
  \end{equation*}
  By Markov's inequality, we have
  \begin{align*}
    \PP^\ell(T_N>t)
    &\le \PP^\ell(w_\ell(t)<1-\delta)\\
    &\le \PP^\ell(2B_\ell^\ell(t)<-v)
    +\PP^\ell\!\left(\sum_{k\ne\ell}e^{2B_k^\ell(t)}
    >\frac{\delta}{1-\delta}e^{4t-v}\right)\\
    &\le \overline\Phi\!\left(\frac{v}{2\sqrt t}\right)
    +\frac{1-\delta}{\delta}e^{-4t+v}
    \sum_{k\ne\ell}\E^\ell e^{2B_k^\ell(t)}\\
    &= \overline\Phi\!\left(\frac{v}{2\sqrt t}\right)
    +\frac{1-\delta}{\delta}(N-1)e^{-2t+v} \,.
  \end{align*}
  Here $B_k^\ell(t)\sim\mathcal N(0,t)$ under $\PP^\ell$, so
  $\E^\ell e^{2B_k^\ell(t)}=e^{2t}$.
  This bound and the fact that
  $\PP(I=\ell)=1/N$ imply
  \begin{align*}
    \PP(T_N>t)
    &=\sum_{\ell=1}^N\PP(I=\ell)\,
    \PP(T_N>t\mid I=\ell)\\
    &=\frac1N\sum_{\ell=1}^N\PP^\ell(T_N>t)\\
    &\le\frac1N\sum_{\ell=1}^N\left[
      \overline\Phi\!\left(\frac{v}{2\sqrt t}\right)
    +\frac{1-\delta}{\delta}(N-1)e^{-2t+v}\right]\\
    &=\overline\Phi\!\left(\frac{v}{2\sqrt t}\right)
    +\frac{1-\delta}{\delta}(N-1)e^{-2t+v} \,,
  \end{align*}
  which is \eqref{eq:late-records}.

  Put $t=\frac12\ln N+u$ and $v=u$ in
  \eqref{eq:late-records}. Since $(N-1)e^{-\ln N}\le1$ and
  $\overline\Phi(x)\le e^{-x^2/2}$ for $x\ge0$, we obtain
  \eqref{eq:late-gauss}.
\end{proof}

\subsection*{Proof of Theorem~\ref{thm:main}}

Fix $r>0$.
We estimate the upper and lower deviations from $\frac12\ln N$ separately.

\emph{Step 1.} For the upper deviation, Proposition~\ref{prop:late-gauss} gives
\begin{equation*}
  \PP(T_N>\frac12\ln N+u)
  \le e^{-u^2/(4\ln N+8u)}+\frac{1-\delta}{\delta}e^{-u}.
\end{equation*}
The first exponent satisfies
\begin{equation*}
  \frac{u^2}{4\ln N+8u}\ge
  \begin{cases}
    u^2/(8\ln N),&0<u\le \frac12\ln N,\\
    u/16,&u>\frac12\ln N.
  \end{cases}
\end{equation*}
Using $x_+=\max\{x,0\}$ and the tail-integral formula,
\begin{align}
  \E \left( T_N-\frac12\ln N \right)_+^r
  &\le r\int_0^\infty u^{r-1}e^{-u^2/(8\ln N)}\,du \notag\\
  &\quad+r\int_0^\infty u^{r-1}
  \left(e^{-u/16}+\frac{1-\delta}{\delta}e^{-u}\right)\,du \notag\\
  &\le C_{r,\delta}((\ln N)^{r/2}+1) \notag\\
  &\le C_{r,\delta}\left(1+\frac1{(\ln2)^{r/2}}\right)(\ln N)^{r/2} \notag\\
  &\le C_{r,\delta}(\ln N)^{r/2} \,.
  \label{eq:upper-centered-moment}
\end{align}
Here we substitute $u=\sqrt{8\ln N}\,z$ in the first integral and use
$\ln N\ge\ln2$.

\emph{Step 2.} For the lower deviation, let $0<u<\frac12\ln N$ and set
$t=\frac12\ln N-u$ and $x=\sqrt{2\ln N}-1$.
Then
\begin{align*}
  x-2\sqrt t
  &=\sqrt{2\ln N}-\sqrt{2\ln N-4u}-1\\
  &=\frac{4u}{\sqrt{2\ln N}+\sqrt{2\ln N-4u}}-1\\
  &\ge\frac{\sqrt2\,u}{\sqrt{\ln N}}-1,
\end{align*}
since the denominator is at most $2\sqrt{2\ln N}$.
Using \eqref{eq:early-records} and the fact that $\overline\Phi$ is
decreasing, we obtain
\begin{align}
  \PP \left( T_N < \frac12 \ln N-u \right)
  &\le\frac1{1-\delta}
  \left[\overline\Phi\!\left(
    \sqrt{2\ln N}-\sqrt{2\ln N-4u}-1\right)
  +\Phi(\sqrt{2\ln N}-1)^{N-1}\right] \notag\\
  &\le\frac1{1-\delta}
  \left[\overline\Phi\!\left(
    \frac{\sqrt2\,u}{\sqrt{\ln N}}-1\right)
  +\Phi(\sqrt{2\ln N}-1)^{N-1}\right] \,.
  \label{eq:lower-centered-tail}
\end{align}

We next bound the second term. For $x>0$, the identity
$\phi'(v)=-v\phi(v)$ and integration by parts give
\begin{align*}
  \overline\Phi(x)
  &=\int_x^\infty\frac{-\phi'(v)}v\,dv\\
  &=\frac{\phi(x)}x
  -\int_x^\infty\frac{\phi(v)}{v^2}\,dv.
\end{align*}
The boundary term at infinity vanishes because $\phi(v)/v\to0$.
Since $v\ge x>0$, we have $1/v^2\le1/x^2$, and hence
\begin{equation*}
  \int_x^\infty\frac{\phi(v)}{v^2}\,dv
  \le\frac1{x^2}\int_x^\infty\phi(v)\,dv
  =\frac{\overline\Phi(x)}{x^2}.
\end{equation*}
Substituting this bound into the subtracted term gives
\begin{equation*}
  \overline\Phi(x)\ge\frac{\phi(x)}x
  -\frac{\overline\Phi(x)}{x^2}.
\end{equation*}
Moving the last term to the left and dividing by $1+x^{-2}$ yields
\begin{equation*}
  \overline\Phi(x)\ge\frac{x}{1+x^2}\phi(x) \geq \frac{\phi(x)}{2x},
  \qquad x\ge1.
\end{equation*}
When $\ln N\ge2$, the choice $x=\sqrt{2\ln N}-1$ satisfies $x\ge1$.
Using $N-1\ge N/2$, we have
\begin{align*}
  (N-1)\overline\Phi \left( \sqrt{2\ln N}-1 \right)
  &\ge
  \frac{e^{\sqrt{2\ln N}-1/2}}
  {4\sqrt{2\pi}\,(\sqrt{2\ln N}-1)}\\
  &\ge
  \frac{e^{\sqrt{2\ln N}}}{8\sqrt{\pi e\ln N}}.
\end{align*}
Since $\Phi(x)=1-\overline\Phi(x)$ and $1-y\le e^{-y}$ for
$0\le y\le1$, the preceding estimate gives, for $\ln N\ge2$,
\begin{align}
  \Phi(\sqrt{2\ln N}-1)^{N-1}
  &=\left[1-\overline\Phi(\sqrt{2\ln N}-1)\right]^{N-1} \notag\\
  &\le\exp\left[-(N-1)\overline\Phi(\sqrt{2\ln N}-1)\right] \notag\\
  &\le\exp\left[-\frac{e^{\sqrt{2\ln N}}}{8\sqrt{\pi e\ln N}}\right].
  \label{eq:Gaussian-maximum-remainder}
\end{align}

Integrating \eqref{eq:lower-centered-tail} and substituting
$u=\sqrt{\frac12\ln N}\,z$ in its first term yields
\begin{align*}
  \E \left( \frac12\ln N-T_N \right)_+^r
  &\le\frac{r(\frac12\ln N)^{r/2}}{1-\delta}
  \int_0^\infty z^{r-1}\overline\Phi(z-1)\,dz \\
  &\quad+\frac{(\frac12\ln N)^r}{1-\delta}
  \Phi \left(\sqrt{2\ln N}-1 \right)^{N-1}  \,.
\end{align*}
On the right hand side,
the integral term is finite.
For the second term, from~\eqref{eq:Gaussian-maximum-remainder}, we have
\begin{align*}
  0\le\frac{(\ln N)^r\Phi \left( \sqrt{2\ln N}-1 \right)^{N-1}}{(\ln N)^{r/2}}
  &=(\ln N)^{r/2}\Phi \left( \sqrt{2\ln N}-1 \right)^{N-1}\\
  &\le\exp\left[
    \frac r2\ln\ln N-\frac{e^{\sqrt{2\ln N}}}{8\sqrt{\pi e\ln N}}
  \right]\longrightarrow0 \,.
\end{align*}
So, $(\ln N)^r\Phi \left( \sqrt{2\ln N}-1
\right)^{N-1}=o((\ln N)^{r/2})$
as $N\to\infty$.
Enlarging $C_{r,\delta}$ to cover the remaining finitely many $N$, we obtain
\begin{equation}
  \label{eq:lower-centered-moment}
  \E \left( \frac12\ln N-T_N \right)_+^r \leq C_{r,\delta}(\ln N)^{r/2} \,.
\end{equation}

\emph{Step 3.} Adding
\eqref{eq:upper-centered-moment} and
\eqref{eq:lower-centered-moment} proves \eqref{eq:moments}.
Dividing by $(\ln N)^r$ gives
\begin{equation*}
  \E\left|\frac{T_N}{\ln N}-\frac12\right|^r
  \le C_{r,\delta} \left( \ln N \right)^{-r/2}\longrightarrow 0  \,,
\end{equation*}
as desired.
\qed

\section{Numerical Simulation}

\subsection{Comparison with earlier simulations}
\label{sec:numerical-comparison}

Ref.~\cite{VJS} used the Euler--Maruyama (EM) method to simulate
the dynamics in \eqref{eq:weight-sde}.
For the Gaussian-noise
simulations, we chose random initial amplitudes to model limited
control over how a single excitation is distributed among the atoms.
For each run, we drew independent standard normal variables
$G_1,\ldots,G_N$ and normalized the vector $(G_1,\ldots,G_N)$.
The initial weights were therefore
$p_n:=w_n(0)=G_n^2/\sum_{j=1}^N G_j^2$, giving $U_n(0)=2p_n-1$.

There are two difficulties with the direct simulation. First, the
diffusion coefficients depend nonlinearly on the weights. While the
model~\eqref{eq:original-model}
preserves positivity of the weights, an EM update can return a
negative weight.
The code discarded nonpositive weights and renormalized the rest.
Both the EM approximation and this correction can bias the
estimated mean collapse time.

Second, the original code formed a dense diffusion matrix at each
step and multiplied it by the vector of Brownian increments. With
$n$ sites, this required $O(n^2)$ storage and work per
step. At $N=10{,}000$, the first matrix alone occupies $800$~MB of RAM.
This is not scalable and is the reason why only  $1{,}000$
atoms with $10{,}000$ realizations
were simulated in Ref.~\cite{VJS}.

\subsection{Revised Numerical Method}

We simulate the Brownian records $Y_n$ in \eqref{eq:Y-definition}. Given the
initial weights $p_n=w_n(0)$, and define
$s_n(t):=\ln p_n+2Y_n(t)$, with $\ln 0=-\infty$.
By \eqref{eq:softmax}, $w_n(t)=e^{s_n(t)}/\sum_{j=1}^N e^{s_j(t)}$.

\begin{algorithm}[!htbp]
  \linespread{1}\selectfont
  \caption{Brownian records for one realization}
  \label{alg:brownian-record}
  \DontPrintSemicolon
  \KwIn{Initial weights $p_1,\ldots,p_N$, time step $h>0$, threshold
  $1-\delta$.}
  \KwOut{First observation time $t$ at which $\max_n w_n(t)\ge 1-\delta$.}
  Draw $I$ with $\PP(I=n\mid p)=p_n$\;
  $t\gets 0$, \quad $Y_n\gets 0$ for $n=1,\ldots,N$\;
  \While{\textnormal{true}}{
    $s_n\gets\ln p_n+2Y_n$ for $n=1,\ldots,N$\;
    $m\gets\max_n s_n$\;
    \lIf{$\sum_{n=1}^N e^{s_n-m}\le (1-\delta)^{-1}$}{\Return $t$}
    \ForPar{$n=1,\ldots,N$}{
      Draw a fresh $\xi_n\sim\mathcal{N}(0,1)$ independently\;
      $Y_n\gets Y_n+(2\one_{\{n=I\}}-1)h+\sqrt{h}\,\xi_n$\;
    }
    $t\gets t+h$\;
  }
\end{algorithm}

Algorithm~\ref{alg:brownian-record} simulates the Brownian record increments in
\eqref{eq:conditional-records}, which is equivalent
to~\eqref{eq:Y-definition} after a change of measure.
Since $w_n=e^{s_n}/\sum_{j=1}^N e^{s_j}$, all weights have the same
denominator. The largest weight therefore occurs at a site where
$s_n=m:=\max_j s_j$. Dividing the numerator and denominator by $e^m$
gives
\begin{equation*}
  \max_n w_n
  = \frac{e^m}{\sum_{j=1}^N e^{s_j}}
  = \frac{1}{\sum_{j=1}^N e^{s_j-m}}.
\end{equation*}
Thus the stopping test is
\begin{equation}
  \max_n w_n\ge 1-\delta
  \quad\Longleftrightarrow\quad
  \sum_{j=1}^N e^{s_j-m}\le \frac{1}{1-\delta}.
  \label{eq:record-stopping}
\end{equation}
The algorithm checks the weights at times $t=0,h,2h,\ldots$ and returns
the first such time at which $\max_n w_n(t)\ge 1-\delta$.

Algorithm~\ref{alg:brownian-record} is simple,
linear, efficient and easy to parallelize.
For each realization, the simulation keeps three
vectors of length $N$. At $N=10{,}000$, these
vectors require only $0.24$~MB, compared with
$800$~MB for the EM method.

\subsection{Results}

We timed both methods with Gaussian initial amplitudes and $h=0.04$
on four threads of an Apple M1 Pro.
At $N=10{,}000$, eight EM realizations took $136.8$ seconds. At this
rate, 288 realizations would take about $82$ minutes. For
Algorithm~\ref{alg:brownian-record}, the median time over three batches
of 288 realizations each was $0.452$ seconds.

As a check, we take $N=2$ with equal initial weights.
Applying It\^o's formula to $(2w_1-1)\ln(w_1/(1-w_1))$ up to $T_2$
gives $\E T_2/\tau_m=(1-2\delta)\ln((1-\delta)/\delta)/4$.
At $\delta=0.005$, the continuous-time mean is $1.31009\,\tau_m$.
At $h=0.04$,
the Brownian record mean was
$1.424\,\tau_m$ and the EM mean was $1.129\,\tau_m$.
Algorithm~\ref{alg:brownian-record}
was closer to the exact mean, with an absolute error of
$0.114\,\tau_m$ compared with $0.181\,\tau_m$ for EM.
Table~\ref{tab:method-comparison} summarizes the efficiency and
accuracy comparisons of the two methods.

\begin{table}[!htbp]
  \caption{Comparison of EM and Brownian records at $h=0.04$ and
  $\delta=0.005$. }
  \label{tab:method-comparison}
  \centering
  \begin{tabular}{@{}lcc@{}}
    \toprule
    & EM & Brownian records \\
    \midrule
    \multicolumn{3}{@{}l}{\textit{$N=10{,}000$, Gaussian initial amplitudes}} \\
    Work per step with $N$ active sites & $O(N^2)$ & $O(N)$ \\
    Array storage & $800$~MB (matrix) & $0.24$~MB (three vectors) \\
    Wall time for 288 realizations & $82$~min (estimated) & $0.452$~s
    (median) \\
    \midrule
    \multicolumn{3}{@{}l}{\textit{$N=2$, equal initial weights}} \\
    Exact continuous-time mean ($\tau_m$) & \multicolumn{2}{c}{$1.31009$} \\
    Sample mean at $h=0.04$ ($\tau_m$) & $1.129$ & $1.424$ \\
    Absolute error in the mean ($\tau_m$) & $0.181$ & $0.114$ \\
    \bottomrule
  \end{tabular}
\end{table}

Figure~\ref{fig:archived-euler} recalls the Gaussian-noise simulations
of Ref.~\cite{VJS}. There we divided the mean collapse times by the
two-site mean $\langle T_2 \rangle$, $0.719304\,\tau_m$, and proposed
a double logarithm fit.
We compare that curve with a logarithmic fit to the same data.

\begin{figure}[htbp]
  \centering
  \includegraphics[width=0.85\linewidth]{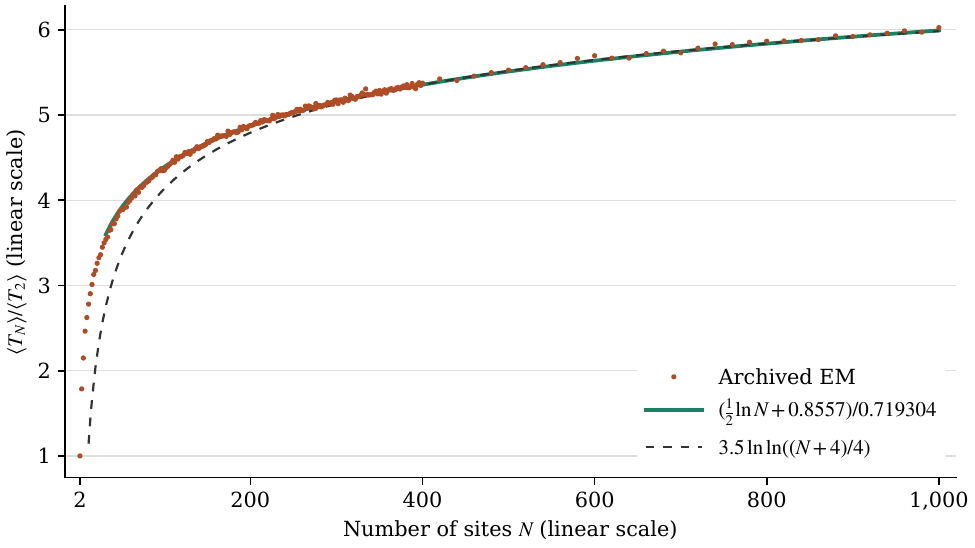}
  \caption{Archived EM means for Gaussian initial amplitudes, divided
    by the two-site mean. The points cover 230 sizes from $N=2$ to $1000$.
    The solid curve is the logarithmic fit $(\tfrac12\ln N+0.8557)/0.719304$
    for $32\le N\le1000$. The dashed curve is the published
  $3.5\ln\ln((N+4)/4)$ for $N>10$.}
  \label{fig:archived-euler}
\end{figure}

\FloatBarrier
Let $m_N$ denote the archived EM mean in units of $\tau_m$.
We fit $g(N)=\tfrac12\ln N+b$ to the 215 archived sizes
$32\le N_i\le1000$, giving each size equal weight. With the slope fixed
at $1/2$, we choose $b$ to minimize the sum of squared errors
\begin{equation*}
  \sum_{i=1}^{215}\left(m_{N_i}-\tfrac12\ln N_i-b\right)^2.
\end{equation*}
Setting the derivative with respect to $b$ to zero gives
$\sum_{i=1}^{215}(m_{N_i}-\tfrac12\ln N_i-b)=0$.
$b$ is the average difference between the observed means and
$\tfrac12\ln N_i$,
\begin{equation*}
  b=\frac{1}{215}\sum_{i=1}^{215}
  \left(m_{N_i}-\tfrac12\ln N_i\right)\simeq0.8557.
\end{equation*}
One sees that the log fit and double log fit are nearly identical
for $N$ between $300$ and $1{,}000$.

Figure~\ref{fig:records-euler} compares the two methods with Gaussian
initial data. We compare the archived EM means with the Brownian
record means used in Fig.~\ref{fig:initial-overlay}.
The Brownian record mean is higher at all nine shared sizes. At $N=2$,
the means are $0.719304\,\tau_m$ for EM and $0.879704\,\tau_m$ for the
Brownian records. At $N=1000$, they are $4.333\,\tau_m$ and
$4.772\,\tau_m$, respectively. The runs also differ in how often they
check the threshold and which time they report.

\begin{figure}[htbp]
  \centering
  \includegraphics[width=0.85\linewidth]{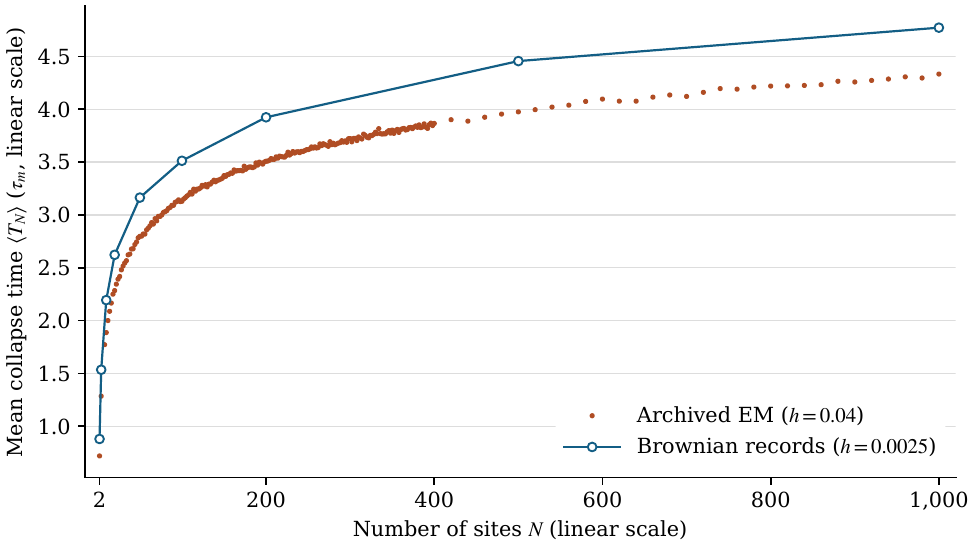}
  \caption{Mean collapse times for Gaussian initial data at threshold
    $1-\delta=0.995$. Both methods used $10{,}000$ realizations per size.
    The archived EM runs used $h=0.04$ and the Brownian record runs
  used $h=0.0025$.}
  \label{fig:records-euler}
\end{figure}

For Fig.~\ref{fig:initial-overlay}, we sample the Brownian records at
16 sizes from $N=2$ to $10^6$ with time step $h=0.0025$ in units of
$\tau_m$. We use at least $1{,}000$ realizations for each size and
each initial condition. Gaussian initial data give a smaller sample
mean at all 16 sizes.

\begin{figure}[htbp]
  \centering
  \includegraphics[width=0.88\linewidth]{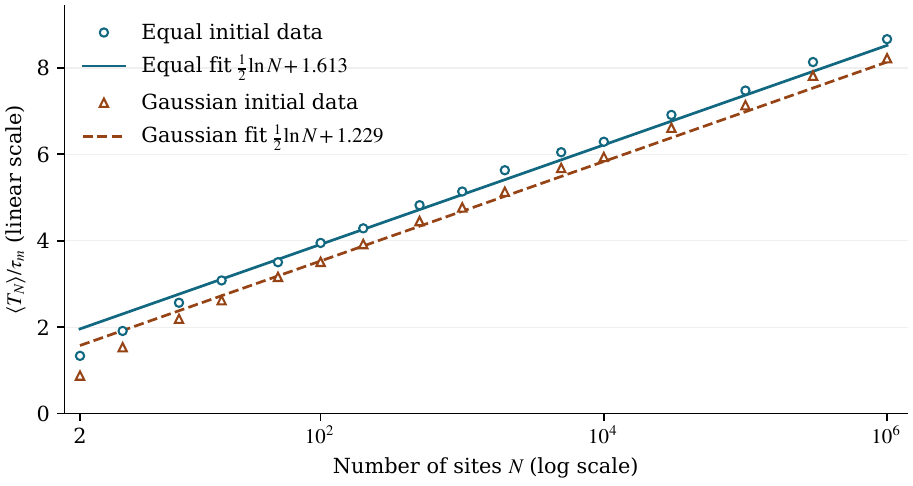}
  \caption{Mean collapse times from Brownian records for equal initial
    weights and Gaussian initial amplitudes. We fit $\tfrac12\ln N+c$
    by least squares, giving each of the 16 sizes from $N=2$ to $10^6$
  equal weight.}
  \label{fig:initial-overlay}
\end{figure}

Figure~\ref{fig:log-comparison} compares the logarithmic and
double-logarithmic fits for equal initial weights. In units of $\tau_m$,
the fits are $\tfrac12\ln N+1.613$ and $2.191+2.067D(N)$, respectively,
where $D(N)=\ln\ln((N+4)/4)$ as in Ref.~\cite{VJS}. We divide the
sample means and both curves by the two-site sample mean on the same
$h=0.0025$, $\langle T_2\rangle=1.33573\,\tau_m$.

\begin{figure}[htbp]
  \centering
  \includegraphics[width=0.88\linewidth]{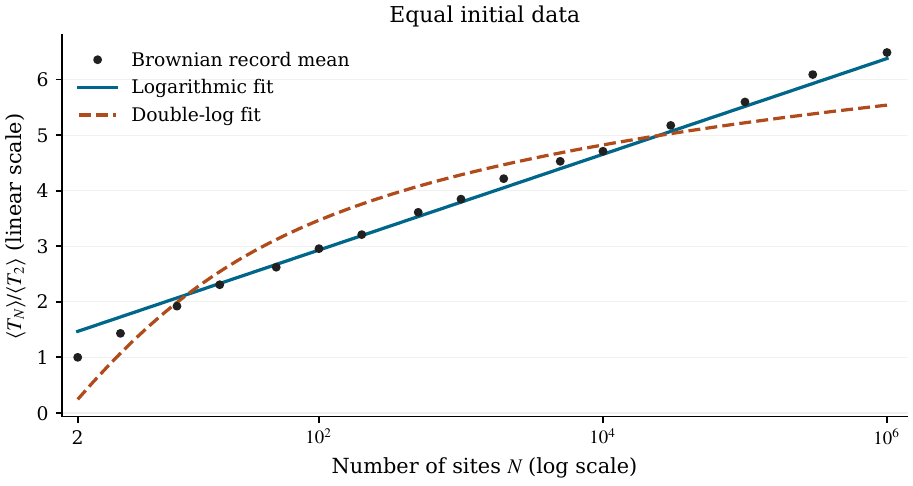}
  \caption{Mean collapse times for equal initial weights, normalized by
  the two-site sample mean. Both fits use all 16 sizes from $N=2$ to $10^6$.}
  \label{fig:log-comparison}
\end{figure}

For the worst case with equal initial
data, the fit in
Fig.~\ref{fig:log-comparison} gives
$\langle T_N\rangle/\langle T_2\rangle\simeq
(\tfrac12\ln N+1.613)/1.33573$, where $\langle T_2\rangle=1.33573\,\tau_m$
is the two-site sample mean at $h=0.0025$.
Extrapolating this fit to $N=10^{80}$ gives $70.16\,\langle T_2\rangle$,
or $93.72\,\tau_m$. At $N=10^{23}$, it gives $21.03\,\langle T_2\rangle$,
or $28.09\,\tau_m$. The growth is slow but unbounded.
Changing the rate of detection $\tau_m$ does not change this overall
scaling law, only the units of time.

\section{Conclusions}
\label{sec:discussion}

The result of this paper is the leading asymptotic law $\E T_N/\ln N\to1/2$
for equal initial weights and a fixed threshold. This corrects the
previously proposed double-logarithmic scaling. The proof conditions
on the eventual site and bounds the early and late crossing
probabilities. Proposition~\ref{prop:reliability} in
Appendix~\ref{app:reliability} gives the exact probability that the
first site to reach the threshold is the eventual outcome.

Numerically, we have found a better simulation method that can capture the
behavior of the system for
very large number of atoms, removing the limited range from the
previous study~\cite{VJS}. It helps definitively demonstrate the
asymptotic behavior proved in this work.

Within the weak measurement and spontaneous collapse models
considered here, a large entangled system reaches a single collapsed
state, in the sense of~\eqref{eq:collapse-time}, in finite time.
The slow logarithmic growth supports the earlier numerical finding
that this time remains short for realistic system sizes. This finite
collapse time distinguishes these models from a many-worlds
description with persistent macroscopic superpositions.

\section*{AI acknowledgement}

TSPV and DWS used ChatGPT 6 Astra (OpenAI) to identify the
$\tfrac12\ln N$ asymptotic law for the mean collapse time.
After the first version of this work appeared on arXiv, Howard Wiseman
drew their attention to WL's work. WL had obtained the same asymptotic
law and proved stronger results. We then agreed to collaborate on a
single paper. ChatGPT 6 Astra helped us combine the drafts and write
the simulation code. We checked the arguments for correctness and
revised them for clarity. We take full responsibility for the content
of this work.
We thank Howard Wiseman for suggesting this collaboration.

\appendix
\section{The weight equation from the stochastic master equation}
\label{app:sme}

We derive the weight equation \eqref{eq:weight-sde} and the Brownian
records \eqref{eq:Y-definition} for the measurement model in
Appendix~A of Ref.~\cite{VJS}.

On $(\mathbb C^2)^{\otimes N}$, we use
$\sigma_z=|1\rangle\langle1|-|0\rangle\langle0|$, with $|1\rangle$
the excited state and $\sigma_z^{(k)}$ acting on site $k$.
The single-excitation subspace is
\begin{equation*}
  \mathcal H_N=\operatorname{span}\{|1\rangle,\dots,|N\rangle\},
  \qquad|n\rangle=|0\cdots0\,1_n\,0\cdots0\rangle.
\end{equation*}
Set $P_n:=|n\rangle\langle n|$, so $\sum_nP_n=\mathrm I_N$, and define
\begin{equation*}
  Z_k:=2P_k-\mathrm I_N=\sigma_z^{(k)}|_{\mathcal H_N}.
\end{equation*}
The mean detector signal is $+1$ on an excited site and $-1$
otherwise~\cite{VJS}.

The local measurement operators and Hamiltonian are
\begin{equation*}
  \mathsf L_k=\frac{\sigma_z^{(k)}}{2\sqrt{\tau_m}},\qquad
  \mathsf H=\frac E2\sum_{k=1}^N\sigma_z^{(k)}
  +\frac\Delta2\sum_{k=1}^N\sigma_x^{(k)}.
\end{equation*}
Compared with Appendix~A of Ref.~\cite{VJS}, this convention reverses
the detector outputs and innovation Brownian motions. The coefficient
$E$ also changes sign when representing the same Hamiltonian.
With $\Delta=0$, as in Ref.~\cite{VJS}, $\mathcal H_N$ is invariant and
$\mathsf H|_{\mathcal H_N}=\frac E2(2-N)\mathrm I_N$, so the Hamiltonian
drops out. Unequal site energies give a diagonal restricted Hamiltonian,
which we remove in the interaction picture. This leaves the occupation
probabilities, Brownian records and coherence magnitudes unchanged.

For an operator $C$, write
\begin{equation*}
  \Diss C(\rho):=C\rho C^*-\tfrac12\bigl(C^*C\rho+\rho C^*C\bigr),\qquad
  \Gain C(\rho):=C\rho+\rho C^*-\Tr\bigl((C+C^*)\rho\bigr)\rho.
\end{equation*}
For self-adjoint $C$ and real $b$, $\Diss{C+b\,\mathrm I}=\Diss C$,
$\Gain{C+b\,\mathrm I}=\Gain C$, $\Diss{bC}=b^2\Diss C$ and
$\Gain{bC}=b\,\Gain C$. Thus
$\mathsf L_k|_{\mathcal H_N}=Z_k/(2\sqrt{\tau_m})
=(P_k-\tfrac12\mathrm I_N)/\sqrt{\tau_m}$ gives the same dissipator
and gain as $P_k/\sqrt{\tau_m}$.

Let $w_k(t):=\Tr(P_k\rho_t)$ and let $S_k$ be the integrated output
of detector $k$. Ideal detection gives
\begin{align}
  dS_k&=(2w_k-1)\,dt+\sqrt{\tau_m}\,dW_k,
  \label{eq:raw-output}\\
  d\rho_t&=\frac1{\tau_m}\sum_{k=1}^N\Diss{P_k}(\rho_t)\,dt
  +\frac1{\sqrt{\tau_m}}\sum_{k=1}^N\Gain{P_k}(\rho_t)\,dW_k,
  \label{eq:sme}
\end{align}
where $W=(W_1,\dots,W_N)$ is the standard $N$-dimensional Brownian
innovation process. The common rank-one eigenprojectors $P_n$ of the
$Z_k$ are stationary for \eqref{eq:sme}. This is the quantum
nondemolition (QND) structure behind state reduction~\cite{LiangReview2026}.

Taking $\Tr(P_i\,\cdot\,)$ in \eqref{eq:sme} and using
\begin{equation*}
  \Tr\bigl(P_i\Diss{P_k}(\rho)\bigr)=0,\qquad
  \Tr\bigl(P_i\Gain{P_k}(\rho)\bigr)
  =2w_i\bigl(\one_{\{k=i\}}-w_k\bigr),
\end{equation*}
we obtain
\begin{equation}
  dw_i=\frac2{\sqrt{\tau_m}}\,w_i\Bigl(dW_i-\sum_{k=1}^Nw_k\,dW_k\Bigr).
  \label{eq:weight-sde-physical}
\end{equation}

Rescale time by setting $t=\tau_mt'$ and
$B_k(t')=W_k(\tau_mt')/\sqrt{\tau_m}$. Then
\eqref{eq:weight-sde-physical} becomes \eqref{eq:weight-sde}.
The rescaled Brownian record $Y_k(t')=S_k(\tau_mt')/\tau_m$ satisfies
\begin{equation*}
  dY_k=U_k\,dt'+dB_k \,,
\end{equation*}
which is \eqref{eq:Y-definition}.

\section{An exponential estimate for state reduction}
\label{app:exponential-reduction}
In this appendix, we improve the result of
Appendix~B of Ref.~\cite{VJS}.
We define the square-root Lyapunov function as in
Refs.~\cite{LAM2019,LA2024}
\begin{equation}
  V_N(w):=\sum_{i \not= j}\sqrt{w_i w_j} =\left(\sum_i\sqrt{w_i}\right)^2-1 \,.
  \label{eq:V-definition}
\end{equation}

\begin{proposition}[Exponential state reduction]
  For every fixed $N \ge 2$, with equal initial weights,
  \begin{equation}
    \E V_N(w(t)) = (N-1)e^{-t}.
    \label{eq:mean-V}
  \end{equation}
\end{proposition}
\begin{proof}
  We first compute the pairwise expectation for general initial weights.
  If $p_i p_j=0$, then $w_iw_j$ vanishes identically by
  \eqref{eq:positive-weights}. For $p_i,p_j>0$, both weights stay positive.
  Applying It\^o's formula to \eqref{eq:weight-sde-vector} gives
  \begin{align*}
    d\sqrt{w_i} &=\frac{dw_i}{2\sqrt{w_i}}
    -\frac{d\langle w_i\rangle_t}{8w_i^{3/2}}\\
    &=\sqrt{w_i}(e_i-w)\cdot dB
    -\frac12\sqrt{w_i}\,\|e_i-w\|^2\,dt \,.
  \end{align*}

  Since $d\langle B_k,B_m\rangle_t=\one_{\{k=m\}}\,dt$, only
  matching Brownian indices contribute. For $i\ne j$,
  \begin{align*}
    d\langle\sqrt{w_i},\sqrt{w_j}\rangle_t
    &=\sqrt{w_iw_j}\sum_k(e_i-w)_k(e_j-w)_k\,dt\\
    &=\sqrt{w_iw_j}\,(e_i-w)\cdot(e_j-w)\,dt.
  \end{align*}

  The product rule now gives
  \begin{align}
    d\sqrt{w_iw_j}
    &=\sqrt{w_iw_j}\,
    \left((e_i-w)\cdot(e_j-w)
    -\tfrac12\|e_i-w\|^2-\tfrac12\|e_j-w\|^2\right)
    \,dt \notag\\
    &\quad+\sqrt{w_iw_j}\,(e_i+e_j-2w)\cdot dB \notag\\
    &=-\sqrt{w_iw_j}\,dt+\sqrt{w_iw_j}
    \left(dB_i+dB_j-2\sum_kw_k\,dB_k\right).
    \label{eq:pair-sde}
  \end{align}

  Applying the product rule to $e^t\sqrt{w_iw_j}$ gives
  \begin{align*}
    d\bigl(e^t\sqrt{w_iw_j}\bigr)
    &=e^t\,d\sqrt{w_iw_j}+e^t\sqrt{w_iw_j}\,dt\\
    &=e^t\sqrt{w_iw_j}\left(dB_i+dB_j-2\sum_kw_k\,dB_k\right).
  \end{align*}

  Since $0\le w_k\le1$, the stochastic coefficients are bounded on each
  finite time interval. Thus $e^t\sqrt{w_iw_j}$ is a martingale, and
  \begin{equation*}
    e^t\E\sqrt{w_i(t)w_j(t)}=\sqrt{p_i p_j}.
  \end{equation*}

  For equal initial weights, $\sqrt{p_i p_j}=1/N$. Summing over the pairs gives
  \begin{equation*}
    \E V_N(w(t))
    =\sum_{i\ne j}\E\sqrt{w_i(t)w_j(t)}
    =N(N-1)\frac{e^{-t}}N=(N-1)e^{-t} \,,
  \end{equation*}
  as desired.
\end{proof}

For $i\ne j$, $\sqrt{w_iw_j}\le(w_i+w_j)/2\le1/2$. Hence
\eqref{eq:mean-V} gives
\begin{equation}
  \E\sum_{i\ne j}w_i(t)w_j(t)
  \le\frac{N-1}{2}e^{-t}.
  \label{eq:exponential-overlap}
\end{equation}
For fixed $N$, this improves the $\cO_N(t^{-1})$ bound in
Appendix~B of Ref.~\cite{VJS}.

\section{Reliability at the first crossing}
\label{app:reliability}

For equal initial weights, let $J$ be the site with the largest
weight at the first
crossing $T_N$,
and let $I$ be the eventual site in \eqref{eq:eventual-site}.
Since $\cF_t^Y=\cF_t$, $\cF_{T_N}$ is the information in the Brownian
records up to $T_N$.
We show that the site that cross the threshold $1-\delta$ will
likely be the eventual excited site, with probability $1-\delta$.

\begin{proposition}[Reliability at the first crossing]
  \label{prop:reliability}
  For $N\ge2$ and $0<\delta<1/2$, the maximizing site $J$ at $T_N$
  is almost surely unique and satisfies $w_J(T_N)=1-\delta$. Moreover,
  \begin{equation*}
    \PP(I=J\mid\cF_{T_N})=1-\delta
    \quad\hbox{almost surely},
  \end{equation*}
  and, hence,
  \begin{equation}
    \PP(I\ne J)=\delta \,.
    \label{eq:reliability}
  \end{equation}
  Furthermore, $\{ I = J \}$ is independent of $T_N$.
\end{proposition}
\begin{proof}
  By \eqref{eq:eventual-site}, $T_N<\infty$ almost surely. Since
  $w_i(0)=1/N<1-\delta$, continuity gives $w_J(T_N)=1-\delta$.
  Since $\sum w_j = 1$ and  $1-\delta>1/2$, $J$ is unique.
  By previous section, each $w_i$ is a bounded martingale with limit
  $\one_{\{I=i\}}$.
  By optional sampling theorem~\cite{KaratzasShreve},
  \begin{equation*}
    w_i(T_N)=\E[\one_{\{I=i\}}\mid\cF_{T_N}]
    =\PP(I=i\mid\cF_{T_N}).
  \end{equation*}

  Therefore,
  \begin{equation*}
    \PP(I=J\mid\cF_{T_N})
    =\sum_i\one_{\{J=i\}}w_i(T_N)=w_J(T_N)=1-\delta \,.
  \end{equation*}
  Taking expectations gives~\eqref{eq:reliability}.

  Lastly, for any $A\in\cF_{T_N}$,
  \begin{equation*}
    \PP(A\cap\{I=J\})
    =\E\!\left[\one_A
    \PP(I=J\mid\cF_{T_N})\right]
    =(1-\delta)\PP(A).
  \end{equation*}
  Thus $\{I=J\}$ is independent of $\cF_{T_N}$ and hence of $T_N$.
\end{proof}

\bibliography{thebib,thebib_liang}

\end{document}